\documentclass[11pt]{article}
\usepackage{odonnell}

\newcites{appb}{Additional references}

\newcommand{\ronote}[1]{}
\newcommand{\rvnote}[1]{}

\newcommand{\dtv}[2]{\mathrm{d}_{\mathrm{TV}}(#1,#2)}

\newcommand{\BC}[2]{\mathrm{BC}(#1,#2)}
\DeclarePairedDelimiterX\diverg[2]{(}{)}{#1 \mathrel{}\mathclose{}\delimsize\|\mathopen{}\mathrel{} #2}

\newcommand{\Dtr}[2]{\mathrm{D}_{\mathrm{tr}}(#1,#2)}

\newcommand{\Fid}[2]{\mathrm{F}(#1,#2)}

\newcommand{\Fidroot}[2]{\mathfrak{f}(#1,#2)}
\newcommand{\Infid}[2]{\ol{\mathrm{F}}(#1,#2)}
\newcommand{\Infidroot}[2]{\ol{\mathfrak{f}}(#1,#2)}

\newcommand{\Succ}{\textsc{Succ}}
\newcommand{\Fail}{\textsc{Fail}}
\newcommand{\Loss}{\textsc{Loss}}

\begin{document}

\title{The Quantum Union Bound made easy}

\author{Ryan O'Donnell\footnotemark[1]
\and Ramgopal Venkateswaran%
\thanks{Computer Science Dept., CMU.
      \{\texttt{odonnell@cs},\texttt{ramgopav@andrew}\}\texttt{.cmu.edu}.
      The author ordering was randomized.}
}
\date{}

\maketitle

\vspace{-2ex}
\begin{abstract}
    \noindent We give a short proof of Gao's Quantum Union Bound and \mbox{Gentle Sequential Measurement theorems.}
\end{abstract}

\section{Introduction}  \label{sec:intro}
Let $\rho \in \C^{d \times d}$ be a quantum mixed state and let $A_1, \dots, A_m \geq 0$ be (orthogonal) projectors on~$\C^d$, which may be thought of as ``quantum events''.
We write $\overline{A}_t = \Id - A_t$, where $\Id$ is the identity operator.
For intuition, we think of the $A_t$'s as ``good'' events that happen with high probability: we write
\[
    {\E}_{\rho}[A_t] \coloneqq \tr(\rho A_t) = 1 - \eps_t,
\]
and hence the ``bad'' event $\ol{A}_t$ has $\E_{\rho}[\ol{A}_t] = \eps_t$.
Suppose we now sequentially measure~$\rho$ with the two-outcome projective measurements $(\ol{A}_1, A_1)$, $(\ol{A}_2, A_2)$, \dots, $(\ol{A}_m, A_m)$.
For $0 \leq t \leq m$, let $\rho_t$ denote the state conditioned on outcomes~$A_1, \dots, A_t$ all occurring.
The Quantum Union Bound question now asks, ``What is the probability, $\Succ$,  that all~$m$ ``good'' outcomes occur?''
We may also ask the related question of Gentle Sequential Measurement: Conditioned on all good outcomes occurring, how far is the resulting state~$\rho_m$ from~$\rho$ (say, in trace distance)?

For full details of the history of these questions, see the discussion in~\cite{KMW19}.
An important milestone regarding the Quantum Union Bound came from Sen~\cite{Sen12}, who established $\Fail \leq 2 \sqrt{\Loss}$, where we denote $\Fail = 1 - \Succ$ and $\Loss = {\sum}_t \eps_t$.
Subsequently, Gao~\cite{Gao15} obtained the square of Sen's upper bound.
His results were:
\begin{theorem}                                     \label{thm:gaos1b}
    \emph{(Gao's Quantum Union Bound.)}  $\Fail \leq 4\Loss$. \hfill $(\heartsuit)$
\end{theorem}
\begin{theorem}                                     \label{thm:gaos1a}
    \emph{(Gentle Sequential Measurement.)} $\Dtr{\rho}{\rho_m} \leq \sqrt{\Loss}$. \hfill $(\diamondsuit)$
\end{theorem}
Khabbazi Oskouei, Mancini, and Wilde~\cite{KMW19} obtained a further improvement to $(\heartsuit)$, discussed in \Cref{sec:refine}.
In this work we give a simple proof of a common generalization of \Cref{thm:gaos1b,thm:gaos1a}.
Denoting  fidelity (\Cref{not:fid}) by $\Fid{{}\cdot{}}{{}\cdot{}}$, we show:
\begin{theorem}                                     \label{thm:1}
    $\displaystyle 1 \leq \sqrt{\Succ} \sqrt{\Fid{\rho}{\rho_m}} + \sqrt{\Fail}\sqrt{\Loss}$.
\end{theorem}
Then to deduce \Cref{thm:gaos1b} from \Cref{thm:1}, we use $\Fid{\rho}{\rho_m} \leq 1$ to get
\begin{equation}
    1 \leq \sqrt{\Succ} + \sqrt{\Fail}\sqrt{\Loss}
    = \sqrt{1-\Fail} + \sqrt{\Fail}\sqrt{\Loss} \ \ \implies\ \  \Fail  \leq \frac{4\Loss}{(1+\Loss)^2} \leq 4 \Loss. \tag*{$(\heartsuit')$}
\end{equation}
Here $\Fail  \leq \frac{4\Loss}{(1+\Loss)^2}$ arises from solving the quadratic for~$\Fail$; it assumes $\Loss \leq 1$.
One can also get $(\heartsuit')$ via AM-GM: $\tfrac12 \Fail \leq 1 - \sqrt{1-\Fail} \leq \sqrt{\Fail}\sqrt{\Loss} = \sqrt{\tfrac12\Fail}\sqrt{2\Loss} \leq \tfrac12(\tfrac12 \Fail + 2\Loss)$.

To deduce \Cref{thm:gaos1a}, we apply Cauchy--Schwarz to \Cref{thm:1} obtaining
\begin{equation}
    1 \leq \cancelto{1}{\sqrt{\Succ + \Fail}}\ \sqrt{\Fid{\rho}{\rho_m} + \Loss} \ \ \implies\ \ %\Fid{\rho}{\rho_m} \geq 1 - \Loss \quad\implies\quad
    \Infid{\rho}{\rho_m} \coloneqq 1 - \Fid{\rho}{\rho_m} \leq \Loss, \tag*{$(\diamondsuit')$}
\end{equation}
stronger than \Cref{thm:gaos1a} thanks to the Fuchs--van de Graaf~\cite{FG99} inequality \mbox{$\Dtr{\rho}{\sigma} \leq \Infid{\rho}{\sigma}^{1/2}$}.

\subsection{Notation}
\begin{notation}  \label{not:fid}
    For two  states $\rho, \sigma \in \C^{d \times d}$, their fidelity is $\Fid{\rho}{\sigma} = \|\sqrt{\rho} \sqrt{\sigma}\|_1^2 = (\tr \sqrt{\sqrt{\rho} \sigma \sqrt{\rho}})^2$, where we recall the Schatten $1$-norm $\|M\|_1 = \tr\sqrt{M M^\dagger} = \tr\sqrt{M^\dagger M}$.
\end{notation}
The fidelity between two states is at most~$1$; this is a consequence of the matrix Cauchy--Schwarz inequality $\|M_1 M_2\|_1^2 \leq \tr(M_1^\dagger M_1) \tr(M_2^\dagger M_2)$.
\begin{notation}    \label{not:condition}
    Let $M$ be a $d$-column matrix with $M^\dagger M \leq \Id$, thought of as a nondestructive measurement matrix (so that $M^\dagger M$ is one element of a POVM).
    The probability of $M$ occurring when~$\rho$ is measured is $\E_\rho[M^\dagger M]$, and we denote the resulting state conditioned on this outcome by
    $
        \rho | M = (M \rho M^\dagger) / \E_\rho[M^\dagger M].
    $
    (We tacitly assume the denominator is nonzero.)
\end{notation}

\begin{remark}
    We work over finite-dimensional Hilbert spaces for simplicity, but this is inessential; the proofs extend to any separable Hilbert space.
\end{remark}

\section{Proof}
\begin{lemma}                                       \label{lem:2b}
    For quantum states $\rho, \sigma \in \C^{d \times d}$ and $A \in \C^{d \times d}$ with $\ol{A} \geq 0$,
    \[
        \sqrt{\Fid{\rho}{\sigma}} \leq \sqrt{{\E}_\sigma[\textnormal{$A^\dagger A$}]} \sqrt{\Fid{\rho}{\sigma|A}} +  \sqrt{{\E}_\sigma[\ol{A}]}\sqrt{{\E}_\rho[\ol{A}]}.% \overset{\textnormal{(if $A$ is a projector)}}{=} \|A\|_\sigma \sqrt{\Fid{\rho}{\sigma|A}} +  \|\ol{A}\|_\sigma \|\ol{A}\|_\rho.
    \]
\end{lemma}
\begin{proof}
    We have $\sqrt{\Fid{\rho}{\sigma}}
    = \|\sqrt{\rho} \sqrt{\sigma}\|_1
    = \|\sqrt{\rho} (A + \ol{A}) \sqrt{\sigma}\|_1
    \leq \|\sqrt{\rho} A \sqrt{\sigma}\|_1 + \|\sqrt{\rho} \ol{A} \sqrt{\sigma}\|_1$.
    On one hand,
    \[
        \|\sqrt{\rho} A \sqrt{\sigma}\|_1 = \tr \sqrt{\sqrt{\rho} A \sigma A^\dagger \sqrt{\rho}} = \sqrt{{\E}_\sigma[A^\dagger A]} \sqrt{\Fid{\rho}{\sigma | A}}.
    \]
    On the other hand, by matrix Cauchy--Schwarz we have
    \[
        \|\sqrt{\rho} \ol{A} \sqrt{\sigma}\|_1 = \left\|\sqrt{\rho} \ol{A}^{1/2} \ol{A}^{1/2} \sqrt{\sigma}\right\|_1\leq \sqrt{{\E}_\sigma[\ol{A}]}\sqrt{{\E}_\rho[\ol{A}]}.
    \]
    (Remark: in \Cref{app:fid} we note that in fact $\|\sqrt{\rho} \ol{A} \sqrt{\sigma}\|_1^2 = \Fid{\rho|\ol{A}^{1/2}}{\sigma|\ol{A}^{1/2}} \cdot {\E}_\sigma[\ol{A}]\cdot {\E}_\rho[\ol{A}]$.)
\end{proof}

For a geometric interpretation with pure states, see \Cref{sec:int2}.  We now prove \Cref{thm:1}.
\begin{proof}
    For $0 \leq t \leq m$, consider the event that the ``good'' outcomes $A_1, \dots, A_t$  all occur.
    We write $p_t$ for the probability of this event, $\rho_t$ for the state $\rho$ conditioned on this event, and $r_t = \sqrt{p_t} \sqrt{\Fid{\rho}{\rho_t}}$.
    For $1 \leq t \leq m$ we write $q_t$ for the probability that $\ol{A}_t$ is the first ``bad'' outcome that occurs. Now
    \[
        r_{t-1} - r_{t}
        = \sqrt{p_{t-1}}\parens*{\sqrt{\Fid{\rho}{\rho_{t-1}}} - \sqrt{{\E}_{\rho_{t-1}}[A_t]} \sqrt{\Fid{\rho}{\rho_{t-1} | A_t}}}
        \leq \sqrt{p_{t-1}}\sqrt{{\E}_{\rho_{t-1}}[\ol{A}_t]}\sqrt{{\E}_{\rho}[\ol{A}_t]}
        = \sqrt{q_t} \sqrt{\eps_t},
    \]
    where the inequality used \Cref{lem:2b} and $A_t^\dagger A_t = A_t$.
    Summing this for $t = 1 \dots m$ yields
    \[
        1 - \sqrt{\Succ} \sqrt{\Fid{\rho}{\rho_m}} = r_0 - r_m \leq \sum_{t=1}^m \sqrt{q_t}\sqrt{\eps_t} \leq \sqrt{{\sum}_t q_t} \sqrt{{\sum}_{t} \eps_t} = \sqrt{\Fail}{\sqrt{\Loss}},
    \]
    where the last inequality is Cauchy--Schwarz.
\end{proof}

Anshu~\cite{Ans21} has observed that if the above proof is written using subnormalized pure states, it becomes structurally very similar to Sen's proof~\cite{Sen12}.

%\subsection*{Acknowledgments}
%R.O.\ would like to thank Costin B\u{a}descu for several discussions.
%R.O.\ is supported by NSF grant FET-1909310 and ARO grant W911NF2110001.
%This material is based upon work supported by the National Science Foundation under the grant number listed above.
%Any opinions, findings and conclusions or recommendations expressed in this material are those of the author    and do not necessarily reflect the views of the National Science Foundation (NSF).

%\bibliographystyle{abbrv}
%\bibliography{quantum-minimal}
%
%\newpage
%
%\appendix

\section{Additional commentary}

\subsection{Simpler proof of Gentle Sequential Measurement}
We remark that if one's only goal is to prove \Cref{thm:gaos1a}, the proof is even simpler.
Assuming $A$ is a projector, applying Cauchy--Schwarz to \Cref{lem:2b} yields
\[
    \sqrt{\Fid{\rho}{\sigma}} \leq \sqrt{{\E}_\sigma[A]} \sqrt{\Fid{\rho}{\sigma|A}} +  \sqrt{{\E}_\sigma[\ol{A}]}\sqrt{{\E}_\rho[\ol{A}]} \leq \cancelto{1}{\sqrt{{\E}_\sigma[A] + {\E}_\sigma[\ol{A}]}} \  \sqrt{\Fid{\rho}{\sigma|A} + {\E}_\rho[\ol{A}]}.
\]
Squaring and rearranging yields:
\begin{proposition}                                     \label{prop:gentle-seq}
    If $\rho, \sigma \in \C^{d \times d}$ are states and $A \in \C^{d \times d}$ is a projector,
    $
        \Infid{\rho}{\sigma|A} \leq \Infid{\rho}{\sigma} + {\E}_\rho[\ol{A}].
    $
\end{proposition}
\noindent Taking $\sigma = \rho_{t-1}$ and $A = A_t$ we get $\Infid{\rho}{\rho_t} \leq \Infid{\rho}{\rho_{t-1}} + \eps_t$, and hence $\Infid{\rho}{\rho_m} \leq \Loss$ by iterating.

\subsection{Fidelity and conditioning}  \label{app:fid}
We first recall some traditional matrix notation:
\begin{notation}
    If $M$ is any matrix, recall that $|M|$ denotes $\sqrt{M^\dagger M}$% (the unique positive semidefinite square-root of $M^\dagger M$)
    , so $\|M\|_1 = \tr |M|$.
    %Note that $|M| = M$ if $M \geq 0$.
\end{notation}
\begin{fact}  \label{fact:abs}
    For any $M \in \C^{m \times \ell}$, $N \in \C^{\ell \times n}$, it is immediate that $|M \cdot N| = \bigl|\,|M| \cdot N\,\bigr|$.
    Taking trace on both sides and using $\|X\|_1 = \|X^\dagger\|_1$, we can infer $\|M \cdot N \|_1 = \bigl\|\,|M| \cdot |N^\dagger| \,\bigr\|_1$.
\end{fact}
Now we introduce some additional notation:
\begin{notation}
    For $\rho \in \C^{d \times d}$ a quantum state and $A \in \mathbb{C}^{d \times d}$, we write $\|A\|_\rho = \sqrt{\E_\rho[A^\dagger A]}$.
\end{notation}
\begin{fact}  \label{fact:Frhorho}
    If $A$ is a projector then $\|A\|^2_\rho = \E_\rho[A] = \Fid{\rho}{\rho | A}$.
    (The latter formula, basically the ``Gentle Measurement Lemma''~\cite{Win99,ON02}, follows just by writing the definitions and using $\sqrt{\rho} A \sqrt{\rho} \geq 0$, since $A \geq 0$.)
\end{fact}

\begin{remark}  \label{rem:notcond2}
     \Cref{not:condition} may alternately be written as  $\displaystyle \sqrt{\rho | M} = \frac{|\sqrt{\rho} M^\dagger|}{||M||_{\rho}}$.
\end{remark}
Although \Cref{thm:1} looks neat as stated, we actually prefer the definition of fidelity that doesn't have the square built in (as in, e.g., the Nielsen--Chuang text).
For lack of better symbols, we introduce the following notation for it:
\begin{notation}
    We write $\Fidroot{\rho}{\sigma} = \sqrt{\Fid{\rho}{\sigma}} = \|\sqrt{\rho}\sqrt{\sigma}\|_1$, and $\Infidroot{\rho}{\sigma} = \sqrt{\Infid{\rho}{\sigma}} = \sqrt{1 - \Fidroot{\rho}{\sigma}^2}$.
\end{notation}

Now the following is an immediate consequence of \Cref{fact:abs} and \Cref{rem:notcond2}:
\begin{proposition} \label{prop:fid-cond3}
    If $\rho, \sigma \in \C^{d\times d}$ are  states and $M, N \in \mathbb{C}^{d \times d}$, then $\displaystyle \Fidroot{\rho | M}{\sigma | N} = \frac{\|\sqrt{\rho} M^\dagger N \sqrt{\sigma}\|_1}{\|M\|_\rho \|N\|_\sigma}$.
\end{proposition}
This formula is quite useful.
In particular ($M = \Id$, $N = A$) it implies $\Fidroot{\rho}{\sigma | A} = \|\sqrt{\rho} A \sqrt{\sigma}\|_1/ \|A\|_\sigma$, which is identical to the first fact derived in our main \Cref{lem:2b}.
Note furthermore that if $A \geq 0$,
\[
     \|\sqrt{\rho} A \sqrt{\sigma}\|_1 = \|\sqrt{\rho} \sqrt{A} \sqrt{A} \sqrt{\sigma}\|_1 = \Fidroot{\rho | \sqrt{A}}{\sigma | \sqrt{A}} \cdot \sqrt{{\E}_\rho[A]}\sqrt{{\E}_\sigma[A]}
\]
where we used \Cref{prop:fid-cond3} again ($M = N = \sqrt{A}$).
This shows the \emph{second} fact derived in our main \Cref{lem:2b} (more precisely, it shows the ``Remark'' at the end, after replacing $A$ with $\ol{A}$).
Finally, putting these two implications together yields:
\begin{corollary}                                       \label{cor:best}
    If $A \geq 0$, then $\displaystyle \Fidroot{\rho}{\sigma | A} \leq \frac{\sqrt{{\E}_\rho[A]}\sqrt{{\E}_\sigma[A]}}{\|A\|_\sigma}$.    If $A$ is furthermore a projector, the right-hand side simplifies to $\sqrt{{\E}_\rho[A]}$; i.e., $\Fidroot{\rho}{\sigma | A} \leq \Fidroot{\rho}{\rho | A}$.
\end{corollary}

\subsection{Obtaining the bound from~\cite{KMW19}} \label{sec:refine}
The proof given by Khabbazi~Oskouei, Mancini, and Wilde~\cite{KMW19} included an improvement to Gao's Quantum Union Bound: they showed that
\begin{equation*} \tag{$\clubsuit$}
    \Fail^* \coloneqq \Fail - \eps_1  \leq p'\eps_1 + (p+p')\sum_{1 < t < m} \eps_t + p \eps_m
\end{equation*}
for any (positive) $p,p'$ with $1/p + 1/p' = 1$.
(Gao's bound is implied by the $p = p' = 2$ case.)
They also gave an application where it is essential that~$p$ may be made arbitrarily close to~$1$.
We can obtain the same bound by slightly modifying our proof of \Cref{thm:1}.

In the modified proof, we simply save on the first term since we know that $q_1 = \eps_1$.
%Writing $\Fail' = \Fail - q_1 = \Fail - \eps_1$ and $\Loss' = \Loss - \eps_1$, this gives
This gives
\[
    1 - \sqrt{\Succ} \sqrt{\Fid{\rho}{\rho_m}} \leq \eps_1 + \sqrt{\Fail^*}\sqrt{\Loss - \eps_1}.
\]
But from \Cref{cor:best} we obtain $\Fid{\rho}{\rho_m} = \Fid{\rho}{\rho_{m-1} | A_m}\leq \Fid{\rho}{\rho | A_m} = 1 - \eps_m$. Thus
\[
    1 - \sqrt{1-\Fail^* -\eps_1} \sqrt{1 - \eps_m} \leq \eps_1 + \sqrt{\Fail^*}\sqrt{{\sum}_{t > 1} \eps_t}.
\]

One can solve the associated quadratic equation for $\Fail^*$ to get a sharp, but messy, bound.\ronote{In particular, I am pretty sure you can get $\Fail^* \leq \frac{2(1-3L)\eps_1 + 4L + 2(1-L)\eps_m}{(1+L)^2}$ for $L = \sum_{1 < t < m} \eps_t$, and also $(\clubsuit)/(1+L)^2$}
More simply, we can use AM-GM twice to get  $1 - \sqrt{1-\Fail^* - \eps_1} \sqrt{1 - \eps_m} \geq \tfrac12 (\Fail^* + \eps_1 + \eps_m)$, and
\[
    \eps_1 + \sqrt{\Fail^*}\sqrt{{\sum}_{t > 1} \eps_t}
        = \eps_1 + \sqrt{p^{-1}\Fail^*}\sqrt{p{\sum}_{t > 1} \eps_t}
        \leq \eps_1 + \tfrac12p^{-1}\Fail^* + \tfrac12 p{\sum}_{t > 1} \eps_t.
\]
Putting these together yields $\displaystyle \Fail^* + \eps_1 + \eps_m \leq 2\eps_1 + p^{-1}\Fail^* + p{\sum}_{t > 1} \eps_t$, which yields $(\clubsuit)$ after multiplication by~$p'$ and rearrangement (note that $p+p' = pp'$).

\subsection{Intuition I: Bhattacharyya coefficient}
A useful way of discovering results concerning quantum fidelity is via analogy with its easier-to-understand classical counterpart:
\begin{notation}
    Recall that for two probability distributions $p,q$ on $[d]$, their Bhattacharyya coefficient is $\BC{p}{q} = \sum_{i=1}^d \sqrt{p_i}\sqrt{q_i} \in [0,1]$.
    (This equals $\Fidroot{\mathrm{diag}(p)}{\mathrm{diag}(q)}$.)
\end{notation}
The well-known classical analogue (indeed, consequence) of the Fuchs--van de Graaf inequality is:
\begin{fact}                                        \label{fact:classical-fvdg}
    The total variation distance $\dtv{p}{q} = \frac12 \|p-q\|_1$ satisfies $\dtv{p}{q} \leq \sqrt{1 - \BC{p}{q}^2}$.
    (This is slightly sharper than bounding total variation distance by Hellinger distance.)
\end{fact}
An event $A \subseteq [d]$ is the analogue of a projector, so the following  can be compared to \Cref{fact:Frhorho}:
\begin{fact}
    If $A \subseteq [d]$ is an event, then $p(A) = \Pr_p[A] = \E_p[1_A] = \BC{p}{p|A}^2$.
\end{fact}
The analogue of our main \Cref{lem:2b} is also natural in the classical case:
\begin{lemma}                                       \label{lem:2b-classical}
    If $A \subseteq [d]$, then $\displaystyle \BC{p}{q}  = \sqrt{q(A)}\,\BC{p}{q|A}+ \sqrt{q(\ol{A})}\sqrt{p(\ol{A})}$.
\end{lemma}
\begin{proof}
    Since $(q|A)_i$ is $q_i/q(A)$ if $i \in A$, and is~$0$ if $i \in \ol{A}$, we get
    \[
        \BC{p}{q} = \sum_{i \in A} \sqrt{p_i} \sqrt{q_i} + \sum_{i \in \ol{A}} \sqrt{p_i} \sqrt{q_i} = \sqrt{q(A)}\, \BC{p}{q|A} + \sum_{i \in \ol{A}} \sqrt{p_i} \sqrt{q_i},
    \]
    and the result follows by applying Cauchy--Schwarz to the second term.
\end{proof}

\subsection{Intuition II: Pure states and geometry} \label{sec:int2}
\newcommand{\Proj}{\mathrm{Proj}}
As observed by Gao~\cite{Gao15}, a purification argument immediately shows that to prove quantum union bounds, it suffices to consider pure states.
This can assist with geometric intuition, particularly if one imagines --- with only mild loss of generality --- that all states and projectors are real.

In this case, let $\ket{\psi_t}$ denote the unit vector in~$\R^d$ obtained by conditioning on the first~$t$ projective measurements succeeding.
Then if $H = H_{t+1}$ denotes the subspace onto which~$A_{t+1}$ projects, the  analysis of the $(t+1)$th measurement really only depends on four vectors, namely $\Proj_{H} \ket{\psi_0}$, $\Proj_{H} \ket{\psi_t}$, $\Proj_{H^\bot} \ket{\psi_0}$, and $\Proj_{H^\bot} \ket{\psi_t}$.
So without loss of generality we may project everything into~$\R^4$, with the first three vectors spanning~$\R^3$.
We can then picture a globe in $\R^3$ of unit radius, with $H_{t+1}$ being  the plane of the equator, $\ket{\psi_0}$ and $\ket{\psi_{t+1}}$ lying on the globe's surface, and $\ket{\psi_t} = r \ket{\wt{\psi}_t} + \ket{\wt{\psi}_t^\bot}$ for some $\ket{\wt{\psi}_t}$ on the globe's surface, with $0 \leq r \leq 1$ and $\ket{\wt{\psi}_t^\bot}$ pointing into the fourth dimension.
For $j \in \{0,t,t+1\}$, we'll write $(\lambda_j, \phi_j)$ for the longitude/latitude of $\ket{\psi_j}$ (or $\ket{\wt{\psi}_j}$ when $j = t$).
We may assume that $\lambda_t = \lambda_{t+1} = 0$, and hence $\ket{\psi_{t+1}} = (0,0)$.
(See the left image in \Cref{fig}.)

For $j \in \{t,t+1\}$, let us write $\Delta_j$ for the angle between $\ket{\psi_0}$ and $\ket{\psi_j}$, and also write $\wt{\Delta}_t$ for the angle between $\ket{\psi_0}$ and $\ket{\wt{\psi}_t}$ (equivalently, $r \ket{\wt{\psi}_t}$).
We claim that
\[
    \cos \Delta_{t+1} = \cos \phi_0  \cos \lambda_0, \qquad \cos \wt{\Delta}_{t} = \cos \phi_t \cos \phi_0 \cos \lambda_0  + \sin \phi_t \sin \phi_0, \qquad \cos \Delta_t \leq \cos \wt{\Delta}_t.
\]
The first formula is the spherical Pythagorean Theorem applied to the  triangle with vertices $\ket{\psi_0}$, $(\lambda_0,0)$, and $\ket{\psi_{t+1}}$.
The second is the great-circle distance formula; equivalently, the spherical Cosine Law applied to the triangle formed by $\ket{\psi_0}$, the north pole (blue dot), and $\ket{\wt{\psi}_t}$.
Finally, the inequality holds because the angle, $\Delta_t$, that $\ket{\psi_0}$ makes with $\ket{\psi_t}$ is at least the angle, $\wt{\Delta}_t$, it makes with~$r\ket{\wt{\psi}_t}$, since the former is equal to the latter plus a vector $\ket{\wt{\psi}_t^\bot}$ that is orthogonal to both $\ket{\psi_0}$ and~$r\ket{\wt{\psi}_t}$.
Combining the above three results now yields
\begin{equation}    \label[ineq]{ineq:gao-trig}
    \cos \Delta_{t} \leq \cos \phi_t  \cos \Delta_{t+1} +  \sin \phi_t \sin \phi_0,
\end{equation}
which\ronote{ignoring signs} is exactly the relationship derived in our main \Cref{lem:2b} (with $\rho$ being $\ket{\psi_0}$ and $\sigma$ being~$\ket{\psi_t}$ and $A$ being projection onto $H_{t+1}$).

\begin{figure}[H]
     \includegraphics[width=.5\textwidth]{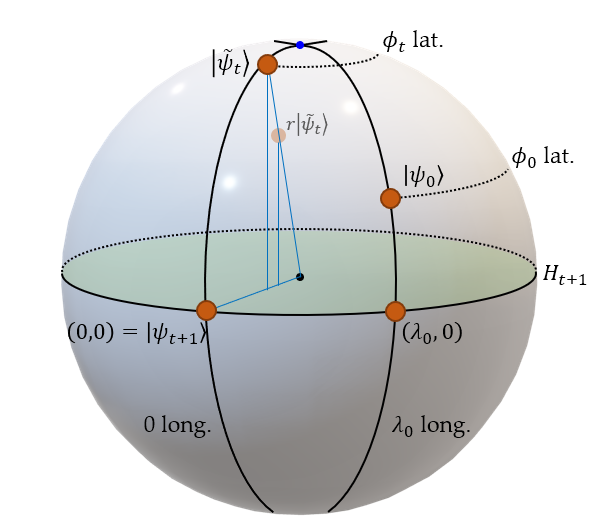}
     \includegraphics[width=.5\textwidth]{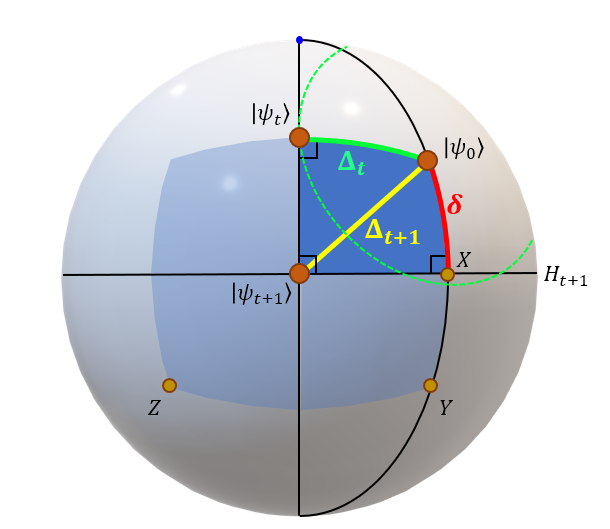}
     \caption{On the left, justifying \Cref{lem:2b} for pure states.  On the right, tightness for \Cref{thm:gaos1a}.} \label{fig}
\end{figure}

\subsection{Tightness}
The factor of~$4$ appearing in the Quantum Union Bound is tight, even in the case of one pure qubit with real amplitudes.
To see this, fix a large~$m$ and then consider $\delta \to 0^+$.
Now suppose the initial state of the qubit is $\ket{0}$, and $A_t$ projects onto the line in~$\R^2$ making an angle of $(-1)^t \cdot \delta$ with $\ket{0}$.
Then one hand, $\eps_t = \sin^2(\pm\delta) \sim \delta^2$ for each~$t$, so $\Loss \sim m \delta^2$.
On the other hand,
\[
    \Fail = 1 - (1 - \sin^2\delta)(1-\sin^2 2\delta)(1 - \sin^2 2\delta) \cdots (1 - \sin^2 2\delta) \sim (4m-3)\delta^2.
\]
From this we see that the constant ``$4$'' in \Cref{thm:gaos1b}'s $\Fail \leq 4\Loss$ cannot be replaced by any smaller constant.

In fact, the same idea can be used to show that the refined bound denoted $(\clubsuit)$ in \Cref{sec:refine}
%\begin{equation*} \tag{$\clubsuit$}
%    \forall c \geq 1 \quad \Fail \leq (2 + c^{-1})\eps_1 + (2 + c + c^{-1})\sum_{1 < t < m} \eps_t + (1 + c)\eps_m \phantom{\quad \forall c \geq 1,}
%\end{equation*}
is asymptotically tight for all fixed~$m \geq 2$ and $p,p'$.
To see this, let $\delta_t = a_t \delta$ for constants $a_1, \dots, a_m$, and let $A_t$ project onto the line in~$\R^2$ making an angle of $(-1)^t \cdot \delta_t$ with $\ket{0}$.
Then on one hand, $\eps_t = \sin^2(\pm \delta_t) \sim a_t^2 \delta^2$, and hence the bound from $(\clubsuit)$ is
\begin{equation}
    \Fail \lesssim \Bigl(a_1^2 + p' a_1^2 + (p+p')\sum_{1 < t < m} a_t^2 + p a_m^2\Bigr)\cdot \delta^2 \tag*{$(\clubsuit')$}
\end{equation}
On the other hand,
\begin{align*}
    \Fail &= 1 - (1- \sin^2\delta_1)(1 - \sin^2(\delta_1 + \delta_2))(1 - \sin^2(\delta_2+\delta_3)) \cdots (1- \sin^2(\delta_{m-1} + \delta_m)) \\
    &\sim \Bigl(a_1^2 + (a_1+a_2)^2 + (a_2 + a_3)^2 + \cdots + (a_{m-1} + a_m)^2\Bigr)\cdot \delta^2.
\end{align*}
But note that whenever $1/p + 1/p' = 1$, it is possible for $a_t, a_{t+1}$ to satisfy $(a_t + a_{t+1})^2 = p'a_t^2 + pa_{t+1}^2$.
(Specifically, this happens if $a_{t+1}/a_t = p'/p$.)
So if this identity is always satisfied, then $(\clubsuit')$ is indeed tight up to lower-order~$O(\delta^4)$ terms.\\

Next we show that Gentle Sequential Measurement bounds $(\diamondsuit')$ are exactly tight (assuming $\sum_t \eps_t \leq 1$), even in the case of pure state qutrits with real amplitudes.
This also implies exact tightness of $(\diamondsuit)$, since $\Dtr{\rho}{\sigma} = \Infid{\rho}{\sigma}^{1/2}$ for pure states $\rho, \sigma$.
To show this, suppose $\ket{\psi_0}$ and $\ket{\psi_t}$ are states in~$\R^3$ at angle $\Delta_t$, and let angle $\delta = \delta_{t+1}$ be given.
(One may imagine that $\sin^2 \Delta_t = \sum_{1 \leq i \leq t} \eps_t$ already, and $\sin^2 \delta = \eps_{t+1}$.)
We will show that there is a two-dimensional subspace $H_{t+1}$ (the image of $A_{t+1}$) such that: (i)~$H_{t+1}$ makes an angle of $\delta$ with $\ket{\psi_0}$; (ii)~the state $\ket{\psi_{t+1}}$ resulting from a successful measurement of~$\ket{\psi_t}$ by~$A_{t+1}$ has an angle $\Delta_{t+1}$ from $\ket{\psi_0}$ satisfying
\begin{equation}    \label{eqn:Deltat}
    \sin^2 \Delta_{t+1} = \sin^2 \Delta_{t} + \sin^2 \delta.
\end{equation}
As we can arrange this for every~$t$, we conclude that $(\diamondsuit')$ can be exactly tight.

It is not hard to see that to maximize $\Delta_{t+1}$, we should choose $H_{t+1}$ to ensure that the (great-circle) arc connecting $\ket{\psi_t}$ to $\ket{\psi_{t+1}}$ is orthogonal to the arc connecting $\ket{\psi_0}$ and $\ket{\psi_t}$, as in the image on the right of \Cref{fig}.
(In that image, one might imagine that $\ket{\psi_{t}}$ could have been any point on the green dashed small circle of radius $\Delta_t$ around~$\ket{\psi_0}$; to maximize $\Delta_{t+1}$ we want the arc connecting $\ket{\psi_{t}}$ to $H_{t+1}$ to be tangent to this green circle.)

Thus it remains to verify that \Cref{eqn:Deltat} holds for the dark blue ``Lambert (three-right-angle) quadrilateral'' with corners $\ket{\psi_0}$, $\ket{\psi_t}$, $\ket{\psi_{t+1}}$, and~$X$ (the state if $\ket{\psi_0}$ were successfully measured by~$A_{t+1}$).
This is an elementary (though perhaps lesser-known) fact of spherical geometry.
To verify it, one may form the three pale blue reflections of the Lambert quadrilateral, giving a centrally symmetric spherical quadrilateral.
Then it is easy to verify that the triangle formed by $\ket{\psi_0}$ and the points depicted as~$Y$ and~$Z$ form a so-called half-sum triangle (a \emph{right} right triangle, in the terminology of~\cite{DS08}), with the triangular angles at $\ket{\psi_0}$ and~$Z$ summing to the angle at~$Y$.
But then \Cref{eqn:Deltat} is immediate from Dickinson and Salmassi's ``Preferred Spherical Pythagorean Theorem''~\cite{DS08}.

\subsection{How we discovered our proof}
The proof we gave is short enough that one might imagine just discovering it from scratch.
Alternatively, one might imagine discovering it by trying to prove the classical Union Bound while working exclusively with Bhattacharyya coefficient.
But in fact, we essentially came up with our proof by iteratively refining and unifying the original proofs of Gao and of Khabbazi~Oskouei--Mancini--Wilde.
(Indeed, along the way we had a version of our proof that was roughly $\alpha$~pages long, for each real number $0.5 \leq \alpha \leq 4.0$.)

The parallels are as follows:
As noted, our main \Cref{lem:2b} essentially becomes the geometric equality \Cref{ineq:gao-trig} when reduced to the pure state case.
In turn, this is equivalent to``inequality~(10)'' in~\cite{Gao15}.
Gao proves this inequality in a different but straightforward fashion, and his deduction of \Cref{thm:gaos1a} from it is also relatively straightforward.
(His ``Lemma~1'' parallels our \Cref{prop:gentle-seq}.)
Then like our proof, Gao's proof of \Cref{thm:gaos1b} is inductive and uses \Cref{ineq:gao-trig} (his ``(10)''), but the inequalities he invokes are significantly more complicated.
It seems that introducing our quantity ``$r_t$'' is important for getting a slick proof.
As for the Khabbazi~Oskouei--Mancini--Wilde proof, the steps in it are all individually straightforward; however, it seems that working explicitly with fidelity, as we do, helps to get a clean proof.
Our key \Cref{lem:2b} may be viewed as hidden in the proof of~\cite[``Lemma~3.3'']{KMW19}; one can extract it upon converting their calculational/iterative proof into an induction.

\subsection*{Acknowledgments}
R.O.\ is supported by NSF grant FET-1909310 and ARO grant W911NF2110001.
We thank Mark Wilde for several remarks that improved the presentation of this paper, and we thank Anurag Anshu for observing the parallel between our proof and Sen's.

\vspace{-2ex}
\bibliographystyle{abbrv}
\bibliography{quantum}

\begin{thebibliography}{1}

\bibitem{Ans21}
A.~Anshu.
\newblock Quantum union bound, 2021.
\newblock \url{https://people.eecs.berkeley.edu/~anuraganshu/Union_bound.pdf}.

\bibitem{DS08}
W.~Dickinson and M.~Salmassi.
\newblock The {\it right} right triangle on the sphere.
\newblock {\em The College Mathematics Journal}, 39(1):24--33, 2008.

\bibitem{FG99}
C.~Fuchs and J.~van~de Graaf.
\newblock Cryptographic distinguishability measures for quantum-mechanical
  states.
\newblock {\em IEEE Trans. Inform. Theory}, 45(4):1216--1227, 1999.

\bibitem{Gao15}
J.~Gao.
\newblock Quantum union bounds for sequential projective measurements.
\newblock {\em Physical Review A}, 92(5):052331, 2015.

\bibitem{KMW19}
S.~Khabbazi{ }Oskouei, S.~Mancini, and M.~Wilde.
\newblock Union bound for quantum information processing.
\newblock {\em Proceedings of the Royal Society A}, 475(2221):20180612, 2019.

\bibitem{ON02}
T.~Ogawa and H.~Nagaoka.
\newblock A new proof of the channel coding theorem via hypothesis testing in
  quantum information theory.
\newblock In {\em Proceedings of the International Symposium on Information
  Theory}, page~73. IEEE, 2002.

\bibitem{Sen12}
P.~Sen.
\newblock Achieving the {H}an--{K}obayashi inner bound for the quantum
  interference channel.
\newblock In {\em Proceedings of the International Symposium on Information
  Theory}, pages 736--740, 2012.

\bibitem{Win99}
A.~Winter.
\newblock Coding theorem and strong converse for quantum channels.
\newblock {\em IEEE Transactions on Information Theory}, 45(7):2481--2485,
  1999.

\end{thebibliography}
%\bibliographystyleappb{alpha}
%\bibliographyappb{quantum}

\end{document}